\documentclass[conference]{IEEEtran}
\ifCLASSINFOpdf
\else
\fi
\usepackage{url}

\usepackage[]{algorithm2e}
\usepackage{algorithmicx}
\usepackage{amssymb}
\usepackage{comment}
\usepackage{color}
\usepackage{marginnote}
\usepackage{amsmath}
\usepackage{amssymb}
\usepackage{amsthm}
\usepackage{multirow}
\usepackage{tikz}
\usetikzlibrary{arrows,automata,shapes,shapes.multipart,decorations.markings,positioning}

\usepackage{subfig}
\usepackage{listings}

\newtheorem{theorem}{Theorem}

\newtheorem{lemma}[theorem]{Lemma}

\newtheorem{problem}{Problem}

\newcommand{\mm}{\mathcal{M}}
\newcommand{\nn}{\mathbb{N}}
\newcommand{\rr}{\mathbb{R}}

\newcommand{\todo}[1]{\marginpar{\textcolor{red}{***}}\textcolor{red}{{\bf TODO:} #1}}

\newcommand{\um}{\mu_{\max}}
\newcommand{\mumax}{\ensuremath{\mu_{\max}}}

\DeclareMathOperator{\inedges}{in}
\DeclareMathOperator{\outedges}{out}
\DeclareMathOperator{\len}{len}

\begin{document}
%
\title{On Systematic Testing for Execution-Time Analysis}

\author{\IEEEauthorblockN{Daniel Bundala}
\IEEEauthorblockA{UC Berkeley\\
Email: bundala@berkeley.edu}
\and
\IEEEauthorblockN{Sanjit A. Seshia}
\IEEEauthorblockA{UC Berkeley\\
Email: sseshia@berkeley.edu}}


%


\maketitle

\begin{abstract}
Given a program and a time deadline, does the program finish before the deadline when executed on a given platform?
With the requirement to produce
a test case when such a violation can occur, we refer to this problem as the worst-case execution-time testing (WCETT) problem. 

In this paper, we present an approach for solving the WCETT problem for loop-free programs by
timing the execution of a program on a small number of carefully
calculated inputs. We then create a sequence of integer linear
programs the solutions of which encode the best timing model
consistent with the measurements. By solving the programs we can find the
worst-case input as well as estimate execution time of any other
input.  Our solution is more accurate than previous approaches and, unlikely
previous work, by increasing the number of measurements we can
produce WCETT bounds up to any desired accuracy. 

Timing of a
program depends on the properties of the platform it executes on. We
further show how our approach can be used to quantify the timing
repeatability of the underlying platform. 
\end{abstract}


%
\IEEEpeerreviewmaketitle


\section{Introduction}
\label{sec:intro}

Execution-time analysis is central to the design and verification
of real-time embedded systems. 
In particular, over the last few decades, much work has been done on 
estimating the worst-case execution time (WCET) (see,
e.g.~\cite{li-book99,wilhelm-tecs07,leeseshia-11}).  
Most of the work on this topic has centered on techniques for finding
upper and lower bounds on the execution time of programs on particular
platforms. 
Execution time analysis is a challenging problem due to the 
interaction between 
the interaction of a large space of program paths with the complexity
of underlying platform (see, e.g.,~\cite{lee-tr07,kirner-isorc08}).
Thus, WCET estimates can sometimes be either too pessimistic
(due to conservative platform modeling) 
or too optimistic (due to unmodeled features of the platform). 

The above challenges for WCET analysis can limit its applicability in
certain settings. One such problem is to verify, given a program $P$, a
target platform $H$, and a deadline $d$, whether $P$ can violate 
deadline $d$ when executed on $H$ --- with the requirement to produce
a test case when such a violation can occur. We refer to this problem
as the worst-case execution-time testing (WCETT) problem. 

Tools that compute
conservative upper bounds on execution time have two limitations in
addressing this WCETT problem:
(i) if the bound is bigger
than $d$, one does not know whether the bound is too loose or whether 
$P$ can really violate $d$, and 
(ii) such tools typically aggregate states for efficiency and hence
do not produce counterexamples.
Moreover, such tools rely on having a fairly detailed timing model of the
hardware platform (processor, memory hierarchy, etc.). In some
industrial settings, due to IP issues, hardware details are not
readily available, making the task much harder for timing analysis
(see e.g, this NASA report for more details~\cite{nasa-toyotaUA-tr11});
in such settings, one needs an approach to timing analysis 
that can work with a ``black-box'' platform.

In this paper, we present an approach to systematically test a
program's timing behavior on a given hardware platform. Our approach
can be used to solve the WCETT problem, in that it not only 
predicts the execution times of the worst-case (longest) program path,
but also produces a suitable test case. It can also be adapted to
produce the top $K$ longest paths for any given $K$. Additionally,
the timing estimate for a program path comes with a ``guard band'' or
``error bound'' 
characterizing the approximation in the estimate --- in all our
experiments, true value was closed to the estimate and well within the
guard band.

Our approach builds upon on prior work on the GameTime 
system~\cite{seshia-iccad08,seshia-acmtecs12}. GameTime allows
one to predict the execution time of any program path without running
it by measuring a small sample of ``basis paths'' and learning a
timing model of the platform based on those measurements. 

The advantage of the GameTime approach is that the platform timing model
is automatically learned from end-to-end path measurements, and thus it is easy to
apply to any platform. However, the accuracy of GameTime's estimate
(the guard bands) depend on the knowledge of a platform parameter $\mumax$
which bounds the cumulative variation in the timing of instructions
along a program path from a baseline value.

For example, a load
instruction might take just $1$ cycle if the value is in a register,
but several $10$s of cycles if it is in main memory and not cached.
The parameter $\mumax$ can be hard to pre-compute based on
documentation of the processors ISA or even its implementation, if
available. 

Our approach shares GameTime's ease of portability to any
platform, and like it, it is also suitable for black-box platforms. 
However, rather than depending on knowledge of $\mumax$, we
show how one can compute the guard bands using an integer linear
programming formulation. Experimental results show that our approach
can be more accurate than the original GameTime
algorithm~\cite{seshia-acmtecs12}, at a small extra computational cost. 
Moreover, our algorithm
is \emph{tunable}: depending on the desired accuracy specified by a
user, the algorithm measures more paths and yields more precise
estimate; possibly measuring all paths if perfect accuracy is
requested.  Finally, we also show how to estimate the parameter $\mumax$.





\section{Preliminaries}
\label{section:model}
Our solution to the problem is an extension of the measurement-based GameTime approach~\cite{original_paper}. We now present the model used in~\cite{original_paper} as well as in this paper. 

\subsection{Model}

To decide whether a program can exceed a given time limit~$d$, it suffices to decide whether the worst input exceeds the limit~$d$. However, recall that without any restrictions on the program, when the program contains unbounded loops or recursion, even determining whether the program terminates not even the number of steps it performs is undecidable. Therefore, we consider only deterministic programs with bounded loops and no recursion. We did not find this limitation restricting as reactive controllers are already often written with this limitation in mind.

Given a computer program, one can associate with it the control-flow graph (CFG) in the usual way (Figure~\ref{fig:cfg}); vertices representing locations and edges the basic blocks (straight-line code fragments) of the program. Since we assume that the loops are bounded and there is no recursion, the loops can be unrolled and all function calls inlined. Thus, the resulting CFG is always a directed and an acyclic graph (DAG).

\begin{figure}
\centering
\lstset{language=C,frame=tb} 
\begin{lstlisting}
int f(int x) {
  if (x % 2 == 0) {
    if (x & 101) {
      x++;
    } else {
      x+=7;
    }
  } 
  return x;
}
\end{lstlisting}

\begin{tikzpicture}[->,>=stealth',bend angle=30,auto]
\tikzstyle{every state}=[inner sep=0mm,circle,outer sep=-3.5mm, draw=none]

  \node (s0) at (0,0) [state, initial above, initial text=]{};
  \node (s1) at (2,-1) [state]{};
  \node (s2) at (1,-2) [state]{};
  \node (s3) at (3,-2) [state]{};
  \node (s4) at (2,-3) [state]{};
  \node (s5) at (0,-4) [state]{};
  \node (s6) at (0,-5) [state]{};   
\path
  (s0) edge [above] node [] {$[x \% 2]$} (s1)
  (s1) edge [right] node [] { $ [x \& 101]$} (s3)
  (s1) edge [left] node [] {$[!(x \& 101)]$} (s2)  
  (s2) edge [left] node [] {$x+=7$} (s4)  
  (s3) edge [right] node [] {$x++$} (s4)  
  (s4) edge [right] node [] {} (s5)  
  (s0) edge [bend right, left] node [] {$[!(x\%2)]$} (s5)  
  (s5) edge [right] node [] {return $x$} (s6)                       
  ;
\end{tikzpicture}
\caption{Source code (top) and its corresponding control flow graph (bottom)}
\label{fig:cfg}
\end{figure}
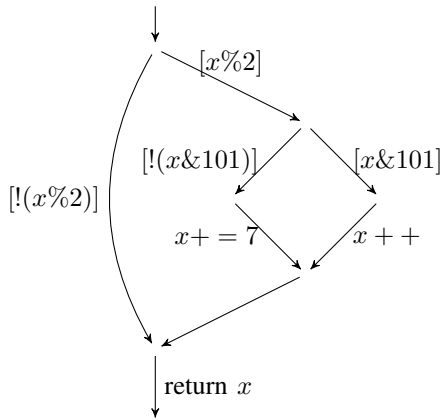

Given a DAG $G = (V,E)$, with the set of vertices $V$, set of edges $E$, we designate two vertices: source~$s$ and sink~$t$ corresponding to the entry and exit points of the program, respectively.

For vertex $v$ of graph, we use $\inedges(v)$ to denote the set of incoming edges to $v$ and $\outedges(v)$ to denote the set of outgoing edges from $v$. 

To model~\cite{original_paper} the execution times of a program, we associate with each edge $e \in E$ cost $w_e$. The cost $w_e$ models the (baseline) execution time of the particular statement the edge $e$ corresponds to. 

As described in the introduction, we measure only execution times of entire source-to-sink paths and not of individual statements. Given a source-to-sink path $x$, the baseline execution time of the path $x$ is $w_x = \sum_{e\in x} w_e$ where the sum is over all edges present in $x$. However, due to caching, branch misses, etc. the actual execution time differs from the baseline and thus the execution time (the \emph{length of the path}) can be modeled as $$w_x = \sum_{e\in x} w_e + d_x$$ where the term $d_x$ denotes the variation from the baseline value. The term $d_x$ is a function of the input and the context the program runs in. This is known as the bounded path-dependent repeatability condition~\cite{wasson}.

The length of path $x$ is denoted $w_x$. Observe that different inputs corresponding to the same path in general take different time to execute. However, we assume that $|d| \leq \mu_x$ for some $\mu_x \in \rr$. We denote the maximum $\max_x \mu_x$ by $\um$. The value $\um$ is a measure of timing repeatability. If $\um=0$ then the system is perfectly time repeatable. In general, the larger the value of $\um$, the less repeatable the timing of the system is.

The aim of this paper is to find a path $x$ such that $w_x$ is maximal. The algorithm in~\cite{original_paper} as well as ours, does not require the knowledge of $\mu_x$'s or $\um$ to find the worst-case execution input. However, the accuracy of the algorithms depend of $\um$, as that is inherent to the timing of the underlying hardware,

Our algorithm carefully synthesizes a collection of inputs on which it runs the given program and \emph{measures} the times it takes to execute each of them. Then, using these measurements, it estimates the length of the longest path. 

Formally, the pair $(x, w_x)$ consisting of a path $x$ and its length $w_x$ is called a \emph{measurement}. We denote the length of the longest path by $w_M$. 

To summarize, throughout the paper we use the following notation.
\begin{itemize}
\item $G$ - underlying DAG
\item $S$ - set of measured paths
\item $\mm = \{ (x_i, l_i) \}$ - set of measurements consisting of pairs path $x_i \in S$ and the observed length $l_i$
\end{itemize}

It was shown in~\cite{original_paper} how, using only $|E|$ measurements of source-to-sink paths, to find an input corresponding to a path of length at least $w_M - 2 |E| \um$. In particular, if the longest path is longer than the second longest path by at least $2 |E| \um$, the algorithm in~\cite{original_paper} in fact finds the longest path. Thus, we say that the ``accuracy'' of the algorithm is $2 |E| \um$. In this paper, we show how to (i) improve the accuracy without increasing the number of measurements, (ii) by increasing the number of measurements improve the accuracy even further, (iii) how to estimate the timing repeatability of the underlying platform (as captured by~$\um$).

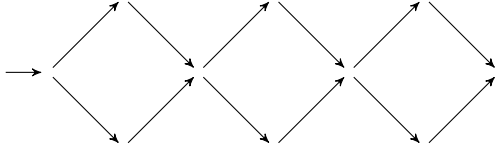
\begin{figure}
\centering
\begin{tikzpicture}[->,>=stealth',bend angle=30,auto]
\tikzstyle{every state}=[inner sep=0mm,circle,outer sep=-3.5mm, draw=none]

  \node (s0a) at (0,0) [state, initial left, initial text=]{};
  \node (s1a) at (1,1) [state]{};
  \node (s2a) at (1,-1) [state]{};
  \node (s3a) at (2,0) [state]{};

  \node (s1b) at (3,1) [state]{};
  \node (s2b) at (3,-1) [state]{};
  \node (s3b) at (4,0) [state]{};

  \node (s1c) at (5,1) [state]{};
  \node (s2c) at (5,-1) [state]{};
  \node (s3c) at (6,0) [state]{};

\path
  (s0a) edge [] node [] {} (s1a)
  (s0a) edge [] node [] {} (s2a)
  (s1a) edge [] node [] {} (s3a)  
  (s2a) edge [] node [] {} (s3a)  

  (s3a) edge [] node [] {} (s1b)
  (s3a) edge [] node [] {} (s2b)
  (s1b) edge [] node [] {} (s3b)  
  (s2b) edge [] node [] {} (s3b)  
  
  (s3b) edge [] node [] {} (s1c)
  (s3b) edge [] node [] {} (s2c)
  (s1c) edge [] node [] {} (s3c)  
  (s2c) edge [] node [] {} (s3c)

  ;
\end{tikzpicture}

\caption{DAG with exponentially many paths}
\label{fig:exp}
\end{figure}

Our algorithm as well as the algorithm in~\cite{original_paper} measures the length of some paths and then estimates the lengths of other paths based on those measurements. Note that as long as not all the lengths of all the paths are measured, the inaccuracy of estimates is unavoidable. Consider for example the graph in Figure~\ref{fig:exp} and assume the graph consists of $n$ ``diamonds''. Clearly, there are $2^n$ source-to-sink paths in the graph.

Assume that $w_e = 1$ for each edge and $\mu_x = 0$ for all paths $x$ except for one path $y$ for which $\mu_y = \um > 0$.  Now, suppose we measure the lengths of some collection of paths~$S$. As long as~$S$ does not contain~$y$, the length of every observed path is $2n$. Hence, any length of $w_y$ of $y$ in the interval $[2n - \um, 2n + \um]$ is consistent with the measurement. Therefore, in the worst case, the best achievable estimate of the length of $w_y$ can always be at least $\um$ from the real answer. 



\footnote{In general, however, it holds that the more paths are included in $S$ the better is the estimate of the longest path.}

We now briefly describe the algorithm in~\cite{original_paper}; we skip the standard technical details such as CFG extraction or how to find an input corresponding to a given path and focus only on how to extract the longest path.

Let $m$ be the number of edges in $E$. Then, by numbering the edges in $E$, one can think of each path $x$ as a vector $p_x$ in $\rr^m$ such that $$p_x(i) = \left\{\begin{array}{rl}0 & \textrm{if ith edge is not used in $x$} \\ 1 & \textrm{if ith edge is used in $x$} \\ \end{array} \right.$$

Now, given two paths $x$ and $y$, one can define the linear combination $a \cdot x + b\cdot y$ for $a,b\in\rr$ in the natural (component wise) way. Thus, one can think of paths as points in an $m$-dimensional vector space over~$\rr$. In particular, it was show in~\cite{original_paper} that there is always a \emph{basis}~$B$ of at most~$m$ source-to-sink paths such that each source-to-sink path $x$ can be written as a linear combination of paths from~$B$: $$p_x = \sum_{b \in B} c_b \cdot p_b$$ where $c_b$'s are coefficients in $\rr$. Moreover, using theory of 2-barycentric spanners~\cite{spanner}, it was shown that $B$ can be chosen in such a way that for any path $x$, it always holds that $|c_b| \leq 2$ for every $b\in B$. The paths in $B$ are called the \emph{basis paths} as they suffice to express any other path. 

Now, if the path $p_x$ can be written as $p_x = \sum_{b \in B} c_b \cdot p_b$ then its estimated (baseline) length is $$w_x = \sum_{b \in B} c_b \cdot w_b$$ where $w_b$'s are the measured lengths of the basis paths. 

The algorithm thus runs the program on the inputs corresponding to the basis paths in order to measure the lengths $p_b$ for each $b\in B$. Moreover, it was shown in~\cite{original_paper} how, by encoding the problem as an integer-linear-program (ILP) instance, to find path $X$ such that the corresponding estimated length: $p_X = \sum_{b \in B} c_b \cdot w_b$ is maximized\footnote{In case the resulting path is infeasible in the program, a constraint is added into the ILP and the ILP is solved again.}.

Consider the accuracy of the estimated length of~$p_X$. By construction, $p_X = \sum_{b \in B} c_b \cdot p_b$. Hence, $$\sum_{e \in X} w_e = \sum_{b \in B} c_b \sum_{e \in b} w_e.$$ Further, for $b \in B$, we have $w_b = \sum_{e \in b} w_e + d_b$. Hence,

\begin{eqnarray*}
\left|w_X - \sum_{b \in B} c_b \cdot w_b\right| & = & \left| \sum_{e \in X} w_e + d_X -  \sum_{b \in B} c_b \cdot w_b\right|\\
& = & \left| \sum_{e \in X} w_e + d_X - \sum_{b \in B} c_b (\sum_{e \in b} w_e + d_b)\right| \\
& = & \big|\sum_{e \in X} w_e - \sum_{b \in B} c_b \sum_{e \in b} w_e \\
& & + d_X - \sum_{b \in B} c_b \cdot d_b\big| \\
& \leq & \left| d_X - \sum_{b \in B} c_b \cdot d_b\right| \\
& \leq & (2 |B| + 1) \um
\end{eqnarray*}

Thus, the algorithm in~\cite{original_paper}, finds the longest path (and the corresponding input) only up to the error term of $(2|B|+1)\um$, under certain assumptions outlined in~\cite{original_paper}. A challenge, as noted earlier, is that it is not easy to estimate the value of $\um$. Consider the first four columns in Table~\ref{table:results}. The second column in the table shows the lengths of the  longest path as estimated by the algorithm in~\cite{original_paper}. However, the third column shows the actual length of the path that is measured when the program is executed on the corresponding inputs. Notice that in some cases the prediction does not match the measured time. Also, the algorithm in~\cite{original_paper} does not provide any error bounds on the estimates, thereby making the predicted values less useful.

In this paper we show how, given the exactly same set of measurements as in~\cite{original_paper}, we can find a tighter estimate and how to incorporate the knowledge of additional measurements to obtain even tighter bounds. In fact, for the benchmarks given in Table~\ref{table:results}, we not only obtain more accurate predictions of running time, we can also give error bounds for our estimates.

\section{Algorithm}
\subsection{Overview}

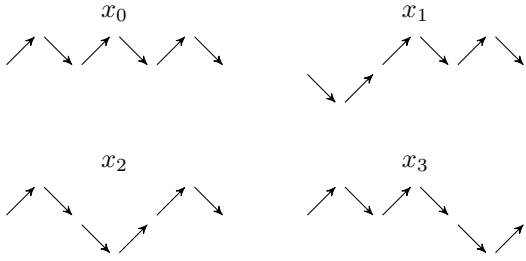
\begin{figure}
\centering
\begin{tikzpicture}[->,>=stealth',bend angle=30,auto,scale=0.5]
\tikzstyle{every state}=[inner sep=0mm,circle,outer sep=-3.5mm, draw=none]

  \node (s0a) at (0,0) [state]{};
  \node (s1a) at (1,1) [state]{};
  \node (s3a) at (2,0) [state]{};

  \node (s1b) at (3,1) [state]{};
  \node (s1bt) at (3,1.5) [state]{$x_0$}{};
  \node (s3b) at (4,0) [state]{};

  \node (s1c) at (5,1) [state]{};
  \node (s3c) at (6,0) [state]{};

  \node (t0a) at (8,0) [state]{};
  \node (t1a) at (9,-1) [state]{};
  \node (t3a) at (10,0) [state]{};

  \node (t1b) at (11,1) [state]{};
  \node (t1bt) at (11,1.5) [state]{$x_1$};
  \node (t3b) at (12,0) [state]{};

  \node (t1c) at (13,1) [state]{};
  \node (t3c) at (14,0) [state]{};

  \node (u0a) at (0,-4) [state]{};
  \node (u1a) at (1,-3) [state]{};
  \node (u3a) at (2,-4) [state]{};

  \node (u1b) at (3,-5) [state]{};
  \node (u1bt) at (3,-2.5) [state]{$x_2$};
  \node (u3b) at (4,-4) [state]{};

  \node (u1c) at (5,-3) [state]{};
  \node (u3c) at (6,-4) [state]{};

  \node (v0a) at (8,-4) [state]{};
  \node (v1a) at (9,-3) [state]{};
  \node (v3a) at (10,-4) [state]{};

  \node (v1b) at (11,-3) [state]{};
  \node (v1bt) at (11,-2.5) [state]{$x_3$};
  \node (v3b) at (12,-4) [state]{};

  \node (v1c) at (13,-5) [state]{};
  \node (v3c) at (14,-4) [state]{};

\path
  (s0a) edge [] node [] {} (s1a)
  (s1a) edge [] node [] {} (s3a)    

  (s3a) edge [] node [] {} (s1b)
  (s1b) edge [] node [] {} (s3b)    
  
  (s3b) edge [] node [] {} (s1c)
  (s1c) edge [] node [] {} (s3c)    

  (t0a) edge [] node [] {} (t1a)
  (t1a) edge [] node [] {} (t3a)    

  (t3a) edge [] node [] {} (t1b)
  (t1b) edge [] node [] {} (t3b)    
  
  (t3b) edge [] node [] {} (t1c)
  (t1c) edge [] node [] {} (t3c)  
  
  (u0a) edge [] node [] {} (u1a)
  (u1a) edge [] node [] {} (u3a)    

  (u3a) edge [] node [] {} (u1b)
  (u1b) edge [] node [] {} (u3b)    
  
  (u3b) edge [] node [] {} (u1c)
  (u1c) edge [] node [] {} (u3c)    

  (v0a) edge [] node [] {} (v1a)
  (v1a) edge [] node [] {} (v3a)    

  (v3a) edge [] node [] {} (v1b)
  (v1b) edge [] node [] {} (v3b)    
  
  (v3b) edge [] node [] {} (v1c)
  (v1c) edge [] node [] {} (v3c)    
  ;
\end{tikzpicture}

\caption{Basis paths for DAG in Figure~\ref{fig:exp}. For example, the path that always takes the bottom path through each diamond can be expressed as $x_1 + x_2 + x_3 - 2 * x_0$}
\label{fig:basis}
\end{figure}

We now give an overview of our algorithm. Recall, that the problem studied can be considered as follows: Given a DAG with source $s$ and sink $t$, find the longest source-to-sink path where the lengths of that paths are modeled as described in the Preliminaries, Section~\ref{section:model}.

The algorithm in~\cite{original_paper} expresses every path as a linear combination of basis paths; using their lengths to estimate the length of the paths not measured. Intuitively, if two paths overlap (they share common edges) then knowing the length of one provides some information about the length of the other. Even basis paths with zero coefficient in the linear combination can provide information about the length of an estimated path. 

In our algorithm, we write integer linear programs (ILPs), with one constraint per measurement, looking for the longest path with the edge weights consistent with the measurements and $\mu_x$. Even though, $\mu_x$ are not observable, we show how to obtain consistent bound on $\um$ from the measurements. 

\subsection{Path Extraction}

In this section we assume that we have a set of measurements~$\mm$, consisting of pairs $(x, l_x)$ where $x$ is a path and $l_x$ is the measured length of $x$ and we show how to find the longest path consistent with the measurements. In Section~\ref{section:basis_computation} we then show how to actually calculate the set~$S$. To make the notation consistent with~\cite{original_paper}, we call the measured paths the \emph{basis paths}, even if they do not necessarily form a basis in the underlying vector space as was the case in~\cite{original_paper}.

Suppose, for the moment, that the value of $\um$ is known and equal to $D \in \rr$. Then the following problem encodes the existence of individual edge weights ($w_e$) such that the cumulative sum ($\sum_{e \in x_i} w_e$) along each measured path is consistent with its measured length; that is, the measured value differs by at most $D$ from the cumulative sum.

\begin{problem}
Input: DAG $G$, a set of measurements $\mm$ and $D~\in~\rr$
$$
\begin{array}{rll}
\max & \len(path) & \\
s.t. & l_i - D \leq \sum_{e \in x_i} w_e \leq l_i + D &
\\ & \quad \textrm{for each measurement $(x_i, l_i) \in \mm$} & \\
vars: &  w_e  \geq 0 & \hspace{-10mm}\textrm{for each edge $e$} \\
\end{array}
$$
\end{problem}

Where $\max \len(path)$ expresses the length of the longest cumulative sum along some source-to-sink path in the graph. We now turn this problem into an ILP by expressing the existence of a path as follows:

\begin{problem}
Input: DAG $G$, a set of measurements $\mm$ and $D~\in~\rr$
$$
\begin{array}{rll}
\max & \sum_{e\in E} p_e & \\
s.t. & l_i - D \leq \sum_{e \in x_i} w_e \leq l_i + D & \\
 & \quad \textrm{for each measurement $(x_i, l_i) \in \mm$} & \\

& \sum_{e \in out(s)} b_e = 1 & \\

& \sum_{e \in in(t)} b_e = 1 & \\

& \sum_{e \in(v)} b_e = \sum_{e \in out(b)} b_e & \\
& \quad\textrm{for each vertex $v \not\in\{s,t\}$} & \\
& p_e \leq w_e & \\
& \quad\textrm{for each edge $e$} & \\
& p_e \leq M \cdot b_e & \\
& \quad\textrm{for each edge $e$} & \\
vars: &   \textrm{for each edge $e$:} & \\
& w_e  \geq 0 & \\
& b_e \in \{0,1\} & \\
& p_e \geq 0 & \\
\end{array}
$$
\label{problem:worst}
\end{problem}

Where $M \in \rr$ is a constant larger than any potential $w_e$. In the implementation, we take $M$ to be the largest $l_i$ in the set of measurements $S$ plus one.

In the above ILP (Problem~\ref{problem:worst}), Boolean variables $b_e$'s specify which edges of the graph are present in the extremal source-to-sink path and $p_e$'s shall equal $b_e \times w_e$. Thus, $\sum_{e\in E} p_e$ denotes the length of the extremal source-to-sink path.

The existence of a path is encoded by the constraints specifying that there is a \emph{flow} from the source to the sink. That is, that exactly one edge from the source has $b_e=1$, exactly one edge to the sink has $b_e=1$ and that for all intermediate vertices, the number of incoming edges to that vertex with $b_e=1$ equals the number of outgoing edges from that vertex with $b_e=1$.

Further, for each edge $e \in E$, we use the variable $p_e$ to denote the multiplication $p_e = b_e \cdot w_e$. As $b_e \in \{0,1\}$, we have $p_e \leq w_e$. Also, the constraints $p_e \leq M \cdot p_e$ ensure that if $b_e=0$ then $p_e = 0$. On the other hand, if $b_e=1$ then constraints imply only that $0 \leq p_e \leq w_e$. Finally, note that the objective function is to maximize $\sum_{e \in E} p_e$. hence, if $b_e=1$ then optimum value for $p_e$ is to set $p_e$ to $w_e$. Hence, in the optimal solution, if $b_e=1$ then $p_e = w_e = 1\cdot w_e = b_e \cdot w_e$ as desired.

Recall, that for measurement $(x_i, l_i)$ it holds that $l_i = \sum_{e \in x_i} w_e + d$ where $|d| \leq \mu_x \leq \um$. Thus, in general, $D$ needs to be at least $0$ to ensure that Problem~\ref{problem:worst} is feasible. By the assumption, taking $D = \um$ yields a feasible ILP.

\begin{lemma}
Problem~\ref{problem:worst} is feasible for $D = \um$.
\end{lemma}

However, the value of $\um$ is neither directly observable nor known as it depends on the actual hardware the program is running on. We now show how to obtain a valid $D$ yielding a feasible ILP in Problem~\ref{problem:worst}. Later we show under what circumstances the solution of the resulting ILP gives the correct longest path.

Consider the following LP\footnote{Note that Problem~\ref{problem:delta} is a linear program and not an integer linear program.}.

\begin{problem}
Input: DAG $G$ and a set of measurements $\mm$
\label{problem:delta}
$$
\begin{array}{rll}
\min & \mu & \\
s.t. & l_i - \mu \leq \sum_{e \in x_i} w_e \leq l_i + \mu & \\
& \quad\textrm{for each measurement $(x_i, l_i) \in \mm$} & \\
vars: &   \textrm{for each edge $e$:} & \\
& w_e  \geq 0 & \\
& \mu  \geq 0\\ 
\end{array}
$$
\end{problem}

Intuitively, the problem above finds the least value of $D$ for which Problem~\ref{problem:worst} is feasible, i.e., the least $D$ consistent with the given set of measurements $M$. Formally, we have: 

\begin{theorem}
Let $p(\mu)$ be the solution of Problem~\ref{problem:delta}. Then taking $D = p(\mu)$ in Problem~\ref{problem:worst} yields a feasible ILP.
\end{theorem}
\begin{proof}
First, note that Problem~\ref{problem:delta} always has a solution, e.g., take $w_e = 0$ and $\mu = \max_i l_i$.

Notice that, assuming there is at least one source-to-sink path, the only possible way for Problem~\ref{problem:worst} to be infeasible is that $D$ is small enough so that the constraints $l_i - D \leq \sum_{e \in x_i} w_e \leq l_i + D$ are inconsistent.

However, by the construction of Problem~\ref{problem:delta}, $p(\mu)$ satisfies, $l_i - p(\mu) \leq \sum_{e \in x_i} w_e \leq l_i + p(\mu)$ for every path $x_i$. The result now immediately follows.
\end{proof}

Note that, by construction, taking $\mu = \um$ in Problem~\ref{problem:delta} is feasible. Hence, $D \leq \um$ as $D$ is the least value consistent with the measurements. 

Notice that the solution of Problem~\ref{problem:delta}, can be used as a measure of timing repeatability of the underlying hardware platform. In case of perfect timing repeatability, that is if each edge (each statement of the underlying program) always took exactly the same time (regardless of concrete values of cache hits, misses, branch predictions, etc) to execute and that $d_x = 0$ for every measured path, then the solution of Problem~\ref{problem:delta} would be~$0$. Conversely, the larger the value of the solution of Problem~\ref{problem:delta} the bigger the discrepancy between different measurements. 

To measure the effect of timing repeatability on the computed value of~$D$, we have taken measurements for a set of benchmarks used to evaluate our tool on (Section~\ref{section:benchmarks}) and randomly perturbed the measured execution times. We have perturbed each measurement randomly by up to $10\%, 25\%$ and $50\%$. Table~\ref{table:deltas} shows the calculated values of~$D$. As expected, the larger the perturbation, the larger the calculated value of~$D$.

\begin{table}
\centering
\caption{Different values of $D$ obtained in the benchmarks by perturbing the measurements.}
\label{table:deltas}

\begin{tabular}{|l||r|r|r|r|}

\hline
\multirow{2}{*}{Benchmark} & \multicolumn{4}{c|}{Perturbation} \\ \cline{2-5}
 & 0\% & 10\% & 25\% & 50\% \\ \hline\hline
altitude & $57.0$ & $66.1$ & $87.8$ & $126.5$ \\ \hline
stabilisation & $343.2$ & $371.3$ & $807.1$ & $1107.2$ \\ \hline
automotive & $1116.0$ & $1281.3$ & $1486.9$ & $2961.3$ \\ \hline
cctask & $73.9$ & $110.6$ & $150.9$ & $270.6$ \\ \hline
irobot & $37.2$ & $95.9$ & $288.8$ & $552.8$ \\ \hline
sm & $0.1$ & $23.2$ & $117.4$ & $216.8$ \\ \hline
\end{tabular}
\end{table} 

\subsection{Optimality}

The solution of Problem~\ref{problem:worst} is the best estimate of the longest path that is consistent with measurements $\mm$. We now show how good the estimate is.

Consider the solution of Problem~\ref{problem:worst} with $D$ equal to the solution of Problem~\ref{problem:delta}. For each edge $e$, let $p(w_e)$ denote the value of the variable $w_e$ in the solution, and let $\tau$ be the path corresponding to the solution of Problem~\ref{problem:worst}. Denote the length of~$\tau$ in the solution by $p(\len(\tau))$. We now show how much $p(\len(\tau))$ differs from the actual length of $\tau$. Specifically, we shall show that the goodness of the estimate of the length of $\tau$ is related to the following ILP\footnote{To find the absolute value $|\len(path)|$ we solve two linear programs. One with the objective function $\max \len(path)$ and one with the objective function $\max -\len(path)$.}. 

\begin{problem}
\label{problem:bound}
Input: DAG $G$ and a set of measurements $\mm$
$$
\begin{array}{rll}
\max & |\len(path)| & \\
s.t. & -1 \leq \sum_{e \in  x_i} w_e \leq +1 & \\
& \quad\textrm{for each measurement $( x_i, l_i) \in \mm$} & \\ 
vars: & w_e & \hspace{-15mm}\textrm{for each edge $e$} \\
\end{array}
$$
\end{problem}

The existence of a path and the length of the path is expressed in the above ILP in exactly the same way as was done in Problem~\ref{problem:worst}. Note that the above ILP is always feasible with $|\len(path)$ at least $1$; one solution is to set $w_e = 1$ for one edge outgoing from the sink and set $w_e = 0$ for all other edges. Further, note that Problem~\ref{problem:bound} depends only on the graph and the set of the measured basis paths; it is independent of the (measured) lengths of the paths. In fact, we can show that as long as some path does not appear in the measurements $\mm$, the solution of the above ILP is strictly greater than $1$.

\begin{theorem}
\label{thm:morethanone}
Let $G$ be a DAG, $\mm$ a set of measurements and $\pi$ a source-to-sink path in $G$ such that $\pi$ is not present in $\mm$. Then the solution of Problem~\ref{problem:bound} is strictly greater than $1$.
\end{theorem}
\begin{proof}
We give a satisfying assignment to variables $w_e$ in Problem~\ref{problem:bound} such that $\len(\pi) > 1$.

Specifically, let $e_i$ be the first edge of $\pi$, that is, the edge outgoing from the sink of $G$. Further, let, $D = \{(u, v) \;|\: u \in \pi\}$ be the set of edges with the initial vertex lying on $\pi$. Then the assignment to weights $w_e$ is as follows:

$$
w_e = \left\{\begin{array}{rr}
			1 + \frac{1}{|E|} & e = e_i \\
			-\frac{1}{|E|} & e \in D \\
			0 & \textrm{otherwise} \\
		\end{array}
\right.
$$  

Note that, with this assignment to $w_e$'s, the length of $pi$ equals $\len(\pi) = 1 + \frac{1}{|E|} > 1$. Now, consider any other path $\tau$ used in measurements $\mm$. In particular, $\tau \neq \pi$. There are $|E|$ edges in $G$ and the weight $w_e$ associated with each edge $e$ is at least $-\frac{1}{|E|}$. Hence, $\len(\tau) \geq |E| \times -\frac{1}{|E|} = -1$. 

Now, if $\tau$ does not include $e_i$, that is $e_i \not\in\tau$ then, $\len(\tau) \leq 0$ as $w_{e_i}$ is the only positive weight. If $\tau$ includes $e_i$, that is $e_i \in \tau$, then $\tau$ necessarilly contains at least one edge from $D$ as $\tau$ is different from $\pi$. Hence, $\len(\tau) \leq 1$. In any case, $-1 \leq \len(\tau) \leq 1$ as required and thus we have given a valid assignment to $w_e$'s with $\len(\pi) > 1$.
\end{proof} 

Recall that the set $S$ denotes the set of paths occurring measurements $\mm$. Let $r(w_e)$ and $r(\mu_ x)$ be the real values of $w_e$ for each edge $e$ and $\mu_x$ for each path $x \in S$. Then for each edge $e$ the expression $|p(w_e) - r(w_e)|$ denotes the difference between the calculated values of $w_e$ and the actual value of $w_e$. Analogously, the expression extends to entire paths: for a path $x$ we have $p(w_x) = \sum_{e\in x}p(w_e)$. Now, the difference for the worst path can be bounded as follows.

\begin{theorem}
\label{thm:bound}
Let $k$ be the solution of Problem~\ref{problem:bound}. Then $$\left|p(len(\tau)) - r(len(\tau))\right| \leq 2k\um$$
\end{theorem}
\begin{proof}
Note that by construction, $r(w_e)$ and $r(\mu_x)$ are a solution of Problem~\ref{problem:delta}. Hence, for every $(x_i, l_i) \in \mm$ it holds that. 

$$l_i - \um \leq l_i - r(\mu_{x_i}) \leq \sum_{e\in x_i} r(w_e) \leq l_i + r(\mu_{x_i}) \leq l_i + \um.$$

Recall that $D \leq \um$ and that that $p(w_e)$ is a solution of Problem~\ref{problem:worst}. Hence, for every $(x_i, l_i) \in \mm$ it holds that:

$$l_i - \um \leq l_i - D \leq \sum_{e\in x_i} p(w_e) \leq l_i + D \leq l_i + \um.$$

Hence, by subtracting the last two equations from each other, we have for any basis path $x_i \in S$ that:

$$- 2 \um \leq \sum_{e\in x_i} p(w_e) - r(w_e) \leq 2 \um$$

Now, dividing by $2 \um$ we have:

$$-1 \leq \frac{\sum_{e\in x_i} p(w_e) - r(w_e)}{2\um} \leq 1.$$
for any basis path $x_i \in S$. 

Thus, the above inequality implies that taking $$w_e = \frac{p_2(w_e) - r(w_e)}{2\um}$$ is a (not necessarily optimal) solution of Problem~\ref{problem:bound}. Since $k$ is the length of the longest path achievable in Problem~\ref{problem:bound}, it follows that for any path $x$ (not necessarily in the basis), we have  $$-k \leq \frac{\sum_{e\in x} p(w_e) - r(w_e)}{2\um} \leq k.$$

By rearranging, we have

$$-2k \um \leq \sum_{e\in x} p(w_e) - r(w_e) \leq 2k \um.$$

In other words, the calculated length differs from the real length by at most $2k\um$, as desired.
\end{proof}

\begin{table}
\caption{Comparison of the accuracy in the longest path extraction between our algorithm and the one in~\cite{original_paper}. For our accuracy, we take $2 * k$ where $k$ is the solution of Problem~\ref{problem:bound}. For~\cite{original_paper} we take $2 * (\textrm{\# basis paths})$.}
\label{table:basis_estimate}
\centering
\begin{tabular}{|l||r|r|r|}
\hline
Benchmark & \# Basis Paths & \cite{original_paper} Accuracy & Our Accuracy \\ 
\hline\hline
altitude & $6$ & $12$ & $10.0$ \\
\hline
stabilisation & $10$ & $20$ & $16.4$ \\
\hline
automotive & $13$ & $26$ & $14.0$ \\ 
\hline
cctask & $18$ & $36$ & $34.0$ \\
\hline
irobot & $21$ & $42$ & $18.0$ \\
\hline
sm & $69$ & $138$ &  $48.6$\\ 
\hline
\end{tabular}
\end{table}

Recall that the algorithm in~\cite{original_paper} has difference between the estimated and the actual length at most $2|E|\um$ whereas our algorithm has $2k\um$ where $k$ is the solution of Problem~\ref{problem:bound}. Observe that the dependence in the error term on $\um$ is unavoidable as $\um$ is inherent to the timing properties of the underlying platform.

For comparison, we have generated the same basis as in~\cite{original_paper} and calculated the corresponding $k$'s for several benchmarks. Table~\ref{table:basis_estimate} summarizes the results (see Table~\ref{table:properties} for the description of benchmarks). As can be seen from the table, when using the same set of measurements, our method gives more accurate estimates than the one in~\cite{original_paper}.

Furthermore, recall that in Problem~\ref{problem:delta} we calculate the best (lower) bound $D$ on $\um$ consistent with the given measurements. Together with the above theorem, this gives ``error bounds'' to the estimate in Problem~\ref{problem:worst}. Specifically, if the length of the longest path computed in Problem~\ref{problem:worst} is $T$ then, the length of the path when measured, that is consistent with the measurements is within $T \pm (2k \times D)$. However, note that this is only the best bound deducible from the measurements since $D \leq \um$. Since $\um$ is not directly observable and we assume no nontrivial bound on $\um$, the length of the path cannot be bounded more accurately without actually measuring the path.

\medskip
The above analysis applies to the extraction of the single longest path. Now, suppose that instead of extracting just one longest path, we want to extract $K$ longest paths. To that end, we iterate the above procedure and whenever a path is extracted, we add a constraint eliminating the path from the solution of Problem~\ref{problem:worst} and then solve the updated ILP. For a path $x$, the constraint eliminating it from the solution space of Problem~\ref{problem:worst} is $\sum_{e\in x} b_e < |x|$. The constraint specifies that not all the edges along $x$ can be taken together in the solution.

As the length of the predicted and the measured length differ, it may happen (e.g., Table~\ref{table:results}) that when measured the length of the path predicted to be the longest is not actually the longest. Thus, to find the longest path, we may need to iterate the above process by generating paths with ever smaller predicted lengths, stopping whenever the current estimate differs by more than $(2k \times D)$ from the longest estimate. 

\subsection{Basis Computation}
\label{section:basis_computation}

The algorithm (Problem~\ref{problem:worst}) to estimate the longest path depends on the set of measurements~$M$ of the basis paths~$S$. In this section we show how to calculate such a set of paths. In general, arbitrary set of paths can be used as basis paths. For example, we have shown in Table~\ref{table:iterative} that using the set of paths used in~\cite{original_paper}, we are able to get more accurate estimates than those obtained in~\cite{original_paper}. 

Recall that (Theorem~\ref{thm:bound}) the accuracy of the solution of Problem~\ref{problem:worst} is tightly coupled with the solution of Problem~\ref{problem:bound}. This leads to a tunable algorithm, which depending on the desired accuracy of the predictions, calculates a set of paths to be used in Problem~\ref{problem:worst}.

Specifically, given a desired accuracy $A \in \rr, A \geq 1$, we want to find a set of feasible paths $S$ such that the solution of Problem~\ref{problem:bound} is at most $A$. We implemented a simple iterative algorithm (Algorithm~\ref{algo:iterative}) that finds such a set by repeatedly extending the set of paths by the path corresponding to the solution of Problem~\ref{problem:bound}.

In the algorithm, if the longest extracted path is infeasible in the underlying program, we add a constraint into the ILP (Problem~\ref{problem:bound}) prohibiting the path. That is, if the longest path is infeasible and $\tau$ is the unsatisfiable core of the longest path\footnote{Minimal set of edges that cannot be taken together as identified by an SMT solver} then we add a constraint that not all the $b_e$'s corresponding to the edges used in $\tau$ are set to~$1$ in Problem~\ref{problem:worst}, i.e., $\sum_{e \in \tau} b_e < |\tau|$ where $|\tau|$ denotes the number of edges in~$\tau$. Then we solve the updated ILP.

\begin{algorithm}
\caption{Iterative algorithm for basis computation}
\label{algo:iterative}
\begin{algorithmic}
\State $S\gets \emptyset$\\
\While{(Solution of Problem~\ref{problem:bound} with paths $S$) $>A$}{
	\State{$x \gets \textrm{longest path in Problem~\ref{problem:bound}}$}\\
	\eIf{$x \textrm{ is feasible}$}{
		\State $S \gets S \cup \{x\}$
}{
		\State{$\textrm{Add a constraint prohibiting $x$}$}
		}
}

\State\Return $S$
\end{algorithmic}
\end{algorithm}

\begin{theorem}
If $A \geq 1$ then the Algorithm~\ref{algo:iterative} terminates with a set $P$ of paths such that the solution of Problem~\ref{problem:bound} with paths $P$ is at most $A$.
\end{theorem}
\begin{proof}
Note that each constraint in Problem~\ref{problem:bound} limits the length of some path to (at most)~$1$.
In particular, if~$S$ contains all the paths in the graph then the solution of Problem~\ref{problem:bound} equals~$1$.

Further, if the algorithm finds some path~$x$ to be the longest in some iteration then the length of~$x$ in all subsequent iterations will be at most~$1$ as $x \in S$. Therefore, as long as the solution is greater than~$1$, the longest path found is different from all the paths found in the previous iterations. 

Also, if the path is infeasible, then a constraint is added that prevents the path occurring in any subsequent iterations. It follows from these considerations that the algorithm keeps adding new paths in each iteration and eventually terminates.

By construction, the solution of Problem~\ref{problem:bound} with the set of paths $S$ is at most $A$.
\end{proof}

In the extreme case of $A=1$, it immediately follows from Theorem~\ref{thm:morethanone}, that Algorithm~\ref{algo:iterative} returns all feasible paths in the underlying graph.

We have implemented the above iterative algorithm and evaluated it on several case studies. Table~\ref{table:iterative} summarizes the number of paths generated by the algorithm for the given accuracy $k$ as well as the running time required to find those paths.

\begin{table}
\caption{Number of paths generated by Algorithm~\ref{algo:iterative} to reach the desired accuracy (second column). Third column shows the accuracy (solution of Problem~\ref{problem:bound}) of the generated set of paths}
\label{table:iterative}
\centering
\begin{tabular}{|l||r|r|r|r|}
\hline
Benchmark & Desired $k$ & Actual $k$ & \# Basis Paths & Time(s) \\
\hline\hline
\multirow{3}{*}{altitude} & 10 &  5.0 & 7 & 0.03 \\
& 5 &  5.0 & 7 & 0.03 \\
& 2 &  1.0 & 10 & 1.66 \\
\hline
\multirow{3}{*}{stabilisation} & 10 &  7.0 & 11 & 0.10 \\
& 5 &  4.7 & 12 & 0.96 \\
& 2 &  2.0 & 40 & 22.78 \\
\hline
\multirow{3}{*}{automotive} & 10 &  7.0 & 14 & 0.14 \\
& 5 &  5.0 & 14 & 0.89 \\
& 2 &  2.0 & 30 & 27.60 \\ 
\hline
\multirow{3}{*}{cctask} & 10 &  9.0 & 19 & 0.20 \\
& 5 &  5.0 & 25 & 4.10 \\
& 2 &  2.0 & 76 & 42.91 \\
\hline
\multirow{3}{*}{irobot} & 10 &  9.0 & 22 & 0.50 \\
& 5 &  5.0 & 34 & 20.13 \\
& 2 &  2.0 & 118 & 182.42 \\
\hline
\multirow{3}{*}{sm} & 22 &  21.8 & 70 & 328.14 \\
& 18 & 18.0  & 73 & 7089.04 \\
& 15 & 14.5  & 77 & 10311.49 \\
\hline
\end{tabular}


\end{table}

We have observed that the basis computation took substantial part of the entire algorithm. However, notice that the basis-computation algorithm (Algorithm~\ref{algo:iterative}) need not start with $S = \emptyset$ and works correctly for any initial collection of paths~$S$. Therefore, as an optimization, we first compute the initial set of paths $S$ using the original algorithm from~\cite{original_paper}, which computes the 2-barycentric spanner of the underlying DAG. Only then we proceed with the iterative algorithm to find a set of paths with the desired accuracy. 

Figure~\ref{fig:kvstime} shows the performance of the iterative algorithm on two benchmarks. The decreasing (blue) line shows how the accuracy $k$ decreases with each path added to the basis. The increasing (red) line shows the time needed to find each path. The figure shows only the performance after the precomputation of the 2-barycentric spanner.

\begin{figure}
\caption{Performance of the Algorithm~\ref{algo:iterative}. The decreasing (blue) line shows the length of $x$ computed ($k$) in line 3. The increasing (red) line shows the time taken to perform each iteration.}
\label{fig:kvstime}
\centering
\subfloat[][cctask]{
\centering
\includegraphics[scale=0.45]{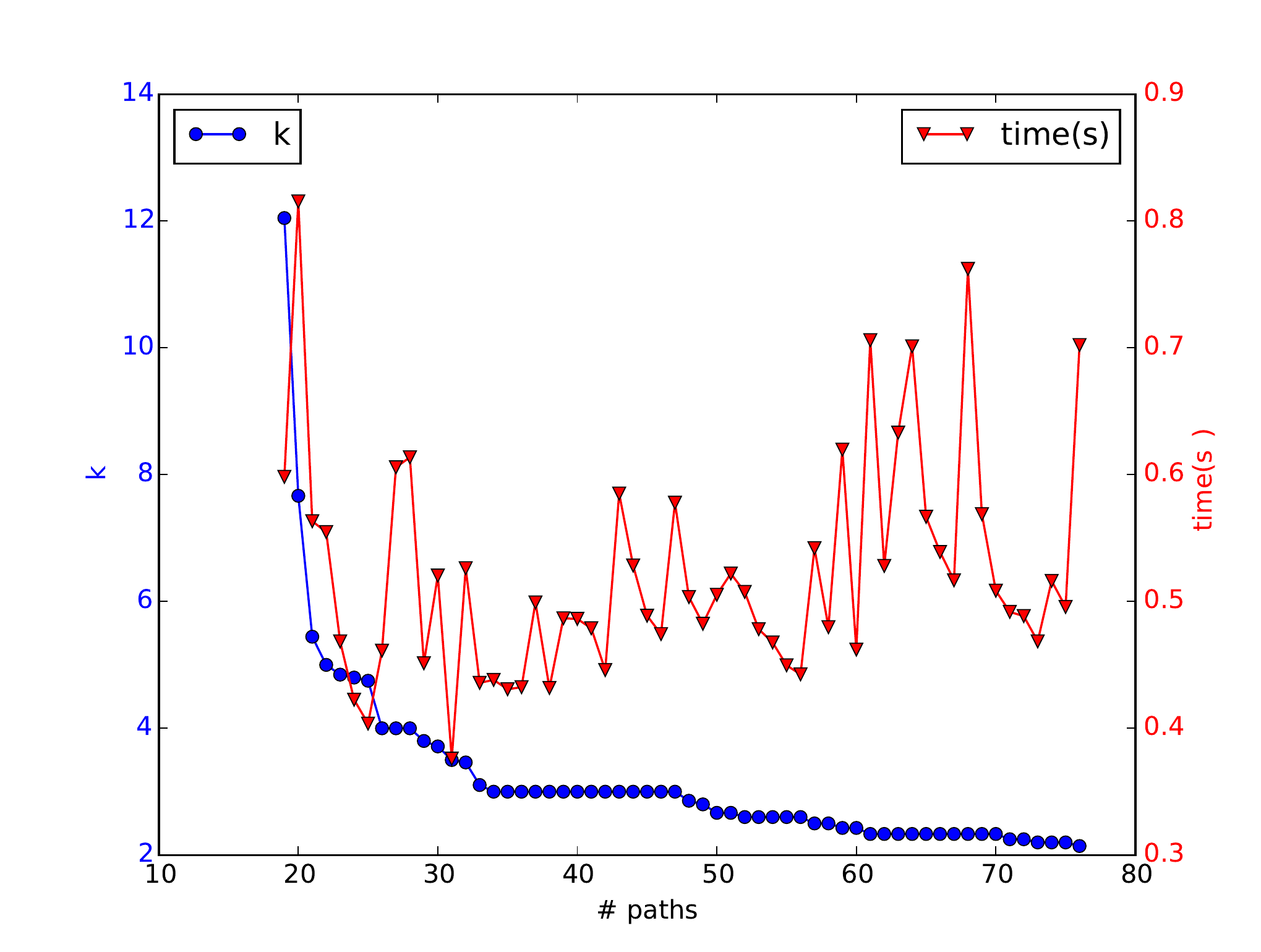}
}

\subfloat[][irobot]{
\centering
\includegraphics[scale=0.45]{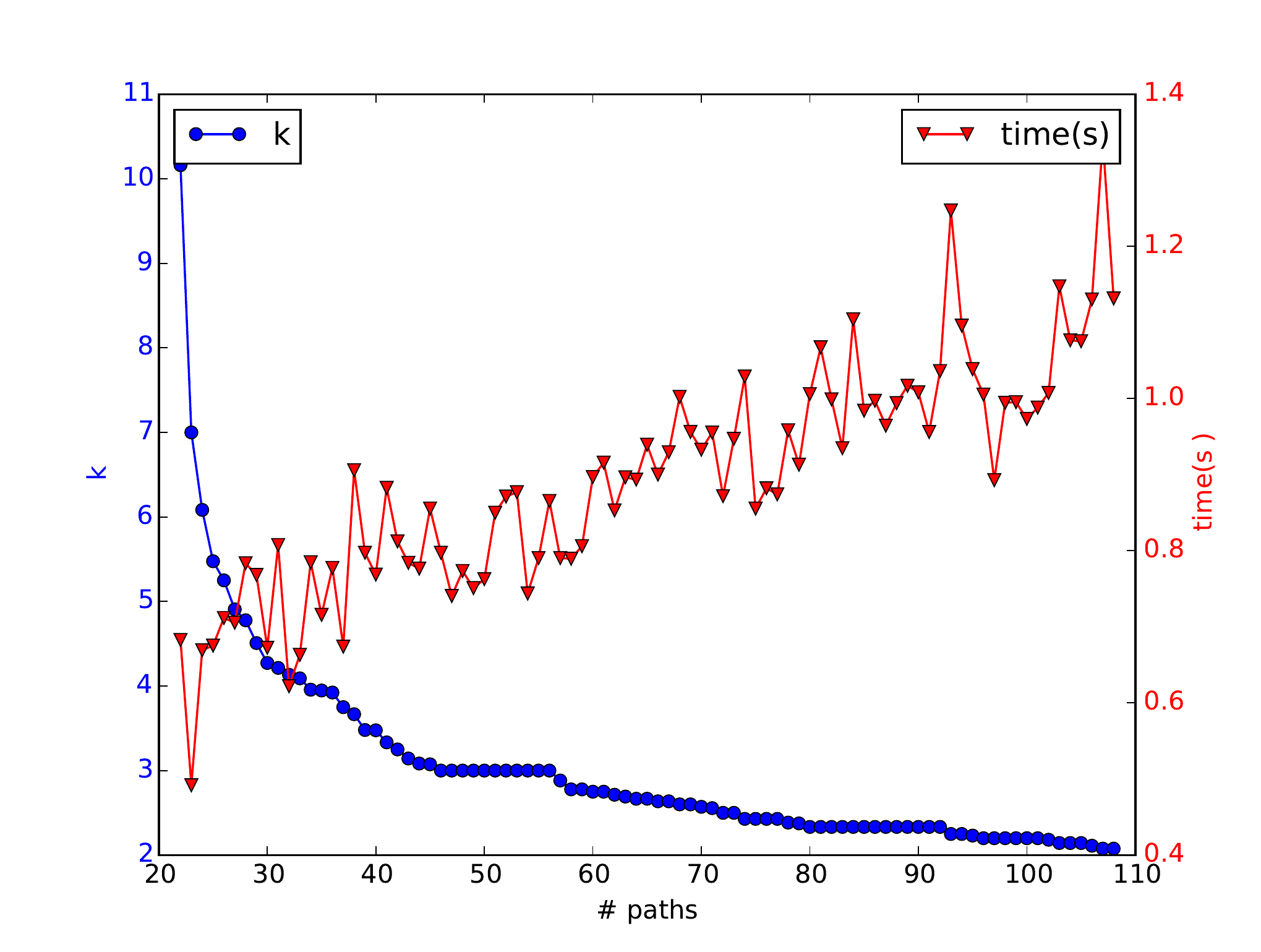}
}

\end{figure}

\section{Evaluation}
\subsection{Implementation}

The algorithm to identify the longest (up to accuracy $k$) path a given program $P$ is shown in Algorithm~\ref{algo:final}.

\begin{algorithm}
\begin{algorithmic}
\State{\textrm{Extract CFG $G$ from the program $P$}}
\State{$S \gets \textrm{basis with accuracy at most $k$ (Algorithm~\ref{algo:iterative})}$}
\State{$D \gets \textrm{Solution of Problem~\ref{problem:delta}}$}
\State{$\tau \gets \textrm{Solution of Problem~\ref{problem:worst} with paths $S$ and $D$}$}
\State\Return{$\textrm{$\tau$ and its estimated length}$} 
\end{algorithmic}
\caption{Algorithm to find the worst-case execution time. Input: Program $P$, accuracy $k$}
\label{algo:final}
\end{algorithm}

We now briefly describe the implementation of main stages of the Algorithm~\ref{algo:final}. Our implementation is build on top of~\cite{original_paper}. See~\cite{original_paper} for further discussion of technical details.

\textbf{CFG extraction} 
The procedure begins by building the CFG $G$ of the given program. The vertices correspond to locations and edges to individual statements. The extraction is build on top of CIL front end for C~\cite{cil}. Note that all ILPs (Problems~\ref{problem:worst},~\ref{problem:delta} and~\ref{problem:bound}) introduce a variable per each edge of $G$ and that each problem optimizes for the length of a source-to-sink path. Thus, if some edge is always followed by another one (the in- and  out-degree of the joining vertex is one) then the edges can be merged into a single one\footnote{For example, in Figure~\ref{fig:exp}, every diamond can be replaced by two edges, one edge for the top half and one edge for the bottom half.} without changing the solution of the ILP problems. Therefore, we process $G$ by merging edges that lie on a single path into a single edge. This reduces the number of variables used and improves the performance.

\textbf{Basis computation} The basis is computed as described in Section~\ref{section:basis_computation}, Algorithm~\ref{algo:iterative}. We use Gurobi solver~\cite{gurobi} to solve individual ILPs and use the 2-barycentric spanner as the initial set $S$. Solving the ILPs posed the main bottleneck of our approach.

\textbf{Input generation} Each path through the CFG~$G$ corresponds to a sequence of operations in the program and hence a conjunction of statements. For a given path, we use the Z3 SMT solver~\cite{z3} to find an input that corresponds to the given path or to prove that the path is infeasible and no corresponding input exists. In the experiments, SMT solving was fairly fast.

\textbf{Longest-path extraction} We solve Problem~\ref{problem:worst} using Gurobi ILP solver~\cite{gurobi}. If the extracted path $\pi$ is infeasible, we add a constraint ($\sum_{e\in\pi} b_e < |\pi|$) eliminating the path from the solution and solve the new problem. Similarly, we solve for $K$ longest paths; we add a constraint prohibiting the extracted path and solve the resulting ILP again, repeating until we successfully generate $K$ feasible paths. 

Recall that the calculated length of the path is accurate only up to the precision $2k \times D$. Thus, we can repeat the process until the length of the extracted path is outside of the range for the longest extracted path. However, in practice we did not find this necessary and found the longest path after a few iterations.

\textbf{Path-Length Measurement} To measure the execution time of the given program on a given input, we create a C program where we set the input to the given values and then run it using a cycle-accurate simulator of the PTARM processor~\cite{ptarm}.

\subsection{Benchmarks}
\label{section:benchmarks}

We have evaluated our algorithm on several benchmarks and compared it with the algorithm in~\cite{original_paper}. We used the same benchmarks as in~\cite{original_paper} as well as benchmarks from control tasks from robotic and automotive settings. The benchmarks in~\cite{original_paper} come from M\"alardalen benchmark suite~\cite{malardalen} and the PapaBench suite~\cite{papabench}.  The authors of~\cite{original_paper} chose implementations of actual realtime systems (as opposed to hand-crafted ones) that have several paths, were of various sizes, but do not
require automatic estimation of loop bounds. 

\begin{table}
\caption{Number of nodes, edges and paths in the CFGs extracted from the benchmarks}
\label{table:properties}
\centering
\begin{tabular}[c]{|l||r|r|r|}
\hline
Benchmark & \# Nodes & \# Edges & \# Paths \\
\hline\hline
altitude & $36$ & $40$ & $11$ \\
\hline
stabilisation & $64$ & $72$ & $216$ \\
\hline
automotive & $88$ & $100$ & $506$ \\
\hline
cctask & $102$ & $118$ & $657$ \\
\hline
irobot & $170$ & $195$ & $8136$ \\
\hline
sm & $452$ & $523$ & $33,755,520$ \\
\hline
\end{tabular}
\end{table}

Since we assume all programs contain only bounded loops and no recursion, we preprocessed the programs by unrolling the loops and inlining the functions where necessary. Table~\ref{table:properties} summarizes properties of the benchmarks used.

Table~\ref{table:results} shows the lengths of five longest paths as found by our algorithm and the one in~\cite{original_paper}. The first half of the table shows the results as obtained by the algorithm in~\cite{original_paper}. The second half shows the results as obtained (together with the ``error bounds'') by our algorithm (Algorithm~\ref{algo:final}). 

Note that in each case the longest path returned by our algorithm is never shorter than the longest path found by~\cite{original_paper}. In half of the cases, our algorithm is able to find a path longer than the one in~\cite{original_paper}. Also notice that the actual measured length is always within the computed ``error bounds''. 

The biggest benchmark, \emph{sm}, is a collection of nested switch-case logic setting state variables but with minimal computations otherwise. Hence, there is a large number of paths ($33,755,520$) in the CFG yet, as expected, the computed $D$ is small.

\begin{table*}[t]
\begin{minipage}{\textwidth}
\centering
\caption{Comparison of the results produced by the algorithm presented in this paper and in~\cite{original_paper} for generating the top five longest paths for a set of benchmarks. Column \emph{Predicted} shows the length predicted by each algorithm. Column \emph{Measured} show the running time measured when run on the corresponding input. For our paper, we give ``error bounds'' of the form $k \times D$. The column \emph{Time} shows the time it takes to generate the estimates and test cases. The largest measured value per each benchmark is shown in bold.}
\label{table:results}

\begin{tabular}{|l||rr|r||rr|r|}
\hline
\multirow{2}{*}{Benchmark} & \multicolumn{3}{c||}{GameTime \cite{original_paper}} & \multicolumn{3}{c|}{Our Algorithm} \\
\cline{2-7}
& Predicted & Measured & Time(s) & Predicted & Measured & Time(s) \\ \hline\hline
\multirow{5}{*}{altitude} & $867$ & $\mathbf{867}$ & $\multirow{5}{*}{11.7}$ & $909\pm 1.0\times57.0$ & $\mathbf{867}$ & $\multirow{5}{*}{14.4}$ \\ \cline{2-3} \cline{5-6}
 & $789$ & $789$ & & $815\pm 1.0\times57.0$ & $758$ &  \\ \cline{2-3} \cline{5-6}
 & $776$ & $751$ & & $732\pm 1.0\times57.0$ & $789$ &  \\ \cline{2-3} \cline{5-6}
 & $659$ & $763$ & & $719\pm 1.0\times57.0$ & $737$ &  \\ \cline{2-3} \cline{5-6}
 & $581$ & $581$ & & $638\pm 1.0\times57.0$ & $581$ &  \\ \cline{2-3} \cline{5-6}
 \hline\hline
 \multirow{5}{*}{stabilisation} & $4387$ & $3697$ & $\multirow{5}{*}{26.4}$ & $4303\pm 2.0\times343.0$ & $3599$ & $\multirow{5}{*}{77.2}$ \\ \cline{2-3} \cline{5-6}
 & $4293$ & $4036$ &  & $4302\pm 2.0\times343.0$ & $\mathbf{4046}$ & \\ \cline{2-3} \cline{5-6}
 & $4290$ & $3516$ &  & $4285\pm 2.0\times343.0$ & $3944$ &  \\ \cline{2-3} \cline{5-6}
 & $4286$ & $3242$ &  & $4284\pm 2.0\times343.0$ & $3516$ &  \\ \cline{2-3} \cline{5-6}
 & $4196$ & $3612$ &  & $4248\pm 2.0\times343.0$ & $3697$ &  \\ \cline{2-3} \cline{5-6}
 \hline\hline
\multirow{5}{*}{automotive} & $13595$ & $8106$ & $\multirow{5}{*}{47.0}$ & $11824\pm 2.0\times1116.0$ & $10982$ & $\multirow{5}{*}{93.8}$ \\ \cline{2-3} \cline{5-6}
 & $11614$ & $9902$ & & $11696\pm 2.0\times1116.0$ & $10657$ &  \\ \cline{2-3} \cline{5-6}
 & $11515$ & $\mathbf{11515}$ & & $11424\pm 2.0\times1116.0$ & $10577$ & \\ \cline{2-3} \cline{5-6}
 & $11361$ & $5010$ & & $11338\pm 2.0\times1116.0$ & $\mathbf{11515}$ &  \\ \cline{2-3} \cline{5-6}
 & $11243$ & $11138$ &  & $9830\pm 2.0\times1116.0$ & $9263$ &  \\ \cline{2-3} \cline{5-6}
\hline \hline
\multirow{5}{*}{cctask} & $991$ & $808$ & $\multirow{5}{*}{29.4}$ & $870\pm 2.0\times73.0$ & $861$ & $\multirow{5}{*}{138.4}$ \\ \cline{2-3} \cline{5-6}
 & $972$ & $605$ & & $869\pm 2.0\times73.0$ & $858$ & \\ \cline{2-3} \cline{5-6}
 & $943$ & $852$ &  & $866\pm 2.0\times73.0$ & $\mathbf{897}$ & \\ \cline{2-3} \cline{5-6}
 & $936$ & $848$ &  & $865\pm 2.0\times73.0$ & $\mathbf{897}$ & \\ \cline{2-3}\cline{5-6}
 & $924$ & $821$ & & $861\pm 2.0\times73.0$ & $873$ & \\ \cline{2-3} \cline{5-6}
 \hline\hline
\multirow{5}{*}{irobot} & $1462$ & $1430$ & $\multirow{5}{*}{50.9}$ & $1451\pm 2.0\times37.0$ & $1406$ & $\multirow{5}{*}{269.0}$ \\ \cline{2-3} \cline{5-6}
 & $1459$ & $1463$ &  & $1450\pm 2.0\times37.0$ & $1411$ &  \\ \cline{2-3} \cline{5-6}
 & $1457$ & $1418$ &  & $1449\pm 2.0\times37.0$ & $1411$ &  \\ \cline{2-3} \cline{5-6}
 & $1454$ & $1451$ &  & $1448\pm 2.0\times37.0$ & $1411$ &  \\ \cline{2-3} \cline{5-6}
 & $1454$ & $1463$ &  & $1448\pm 2.0\times37.0$ & $\mathbf{1464}$ & \\ \cline{2-3} \cline{5-6}
 \hline\hline
\multirow{5}{*}{sm} & $2553$ & $\mathbf{2550}$ & \multirow{5}{*}{$211.0$} & $2552\pm 21.81\times0.2$ & $\mathbf{2550}$ & \multirow{5}{*}{$3290.2$} \\ \cline{2-3} \cline{5-6}
 & $2551$ & $\mathbf{2550}$ &  & $2551\pm 21.81\times0.2$ & $\mathbf{2550}$ &  \\ \cline{2-3} \cline{5-6}
 & $2536$ & $2537$ &  & $2536\pm 21.81\times0.2$ & $2537$ &  \\ \cline{2-3} \cline{5-6}
 & $2534$ & $2537$ &  & $2536\pm 21.81\times0.2$ & $2537$ &  \\ \cline{2-3} \cline{5-6}
 & $2532$ & $2537$ &  & $2531\pm 21.81\times0.2$ & $2537$ &  \\ \cline{2-3} \cline{5-6}
\hline

\end{tabular}
\end{minipage}
\end{table*}

\section{Conclusion}

In this paper, we have addressed the problem of estimating the worst-case timing of a program via systematic testing
on the target platform. Our approach not only generates an estimate of worst-case timing, but can also produces test
cases showing how that timing is exhibited on the platform. 
Our approach improves the accuracy of the previously published GameTime approach, while also
providing error bounds on the estimate. 

Note that our approach can be adapted to produce timing estimates along arbitrary program paths.
In order to do this, one can fix variables $b_e$ in Problem~\ref{problem:worst} suitably. 
Thus, we can also estimate the longest execution of a given path that is consistent with the measurements.

In the paper we have analyzed the timing behavior of a given program. However, instead of measuring cycles we can measure energy consumption of the program executions. The same techniques can then be applied to find the input consuming the most energy. In general, the approach presented in this paper can also be extended to other quantitative properties of the program and is not limited only to the WCETT analysis.

\bibliographystyle{IEEEtran}
\bibliography{bib,gametime}

%

%


\end{document}